\documentclass[a4paper,UKenglish,cleveref,autoref,thm-restate]{lipics-v2021}
\usepackage{graphicx} % Required for inserting images
\usepackage{todonotes}
\usepackage{nicefrac}
\usepackage{xspace}
\usepackage[basic]{complexity}
\usepackage{cite}
\usepackage{mathtools}

\newcommand{\FPT}{\ComplexityFont{FPT}}
\newcommand{\XP}{\ComplexityFont{XP}}

\newcommand{\MSBM}{{\normalfont\textup{\textsc{SBMP}}}}
\newcommand{\minMSBM}{{\normalfont\textup{\textsc{min-SBMP}}}}

\newcommand{\OPT}{\ensuremath{\text{\normalfont\textsc{Opt}}}\xspace}
\newcommand{\tw}{\ensuremath{\mathrm{tw}}}

\newcommand{\WMFS}{\textup{\textsc{Weighted Monotone Formula Satisfiability}}}
\newtheorem{problem}{Problem}
\newclass{\paraNP}{para\text{-}NP}

\usetikzlibrary{patterns.meta}

\title{Matching with Multiple Bottlenecks: Parameterized Complexity and Approximation}

\author{Jonas Friemel}{Bochum University of Applied Sciences, Germany}{jonas.friemel@hs-bochum.de}{https://orcid.org/0009-0009-6270-4779}{Funded by the Deutsche Forschungsgemeinschaft (DFG, German Research Foundation) -- 530918134 (project SpaceAnts).}
\author{Tilo Hoitz}{TU Berlin, Germany}{hoitz@campus.tu-berlin.de}{https://orcid.org/0009-0008-2083-5466}{}
\author{Phillip Keldenich}{TU Braunschweig, Germany}{keldenich@ibr.cs.tu-bs.de}{https://orcid.org/0000-0002-6677-5090}{}
\author{Arne Schmidt}{TU Braunschweig, Germany}{aschmidt@ibr.cs.tu-bs.de}{https://orcid.org/0000-0001-8950-3963}{}

\authorrunning{J. Friemel, T. Hoitz, P. Keldenich and A. Schmidt}

\Copyright{Jonas Friemel, Tilo Hoitz, Phillip Keldenich and Arne Schmidt}

 %TODO mandatory; please add comma-separated list of keywords

\category{} %optional, e.g. invited paper

\relatedversion{}

\nolinenumbers %uncomment to disable line numbering

\hideLIPIcs

\begin{document}

\maketitle
\begin{abstract}
	We consider a matching problem in which the cost of each edge is a vector with $k$ components.
	The cost of a matching is the sum of the bottlenecks over all components, and we ask whether there is a perfect matching of cost at most some value $Z$.
	This type of matching has applications in heavily synchronized job-shop scheduling problems and in reconfiguration problems,
	where movement is restricted to a single direction per step.
	In this paper, we analyze the problem from a parameterized complexity perspective and provide various results including $\FPT$-membership for parameters $k$ and $Z$ combined,
	as well as $\W[\P]$-membership and $\W[\SAT]$-hardness for each of the two parameters individually.
	The reduction also implies $\paraNP$-hardness parameterized by either maximum degree or treewidth.
	We further show hardness of approximation within a super-logarithmic factor for the optimization variant and provide a $\nicefrac k d$-approximation algorithm for any constant $d \leq k$ as well as an efficient approximation scheme parameterized by $k$.
	With parameter $Z$, we show that no $\FPT$-time $F(Z)$-approximation algorithm is possible for any computable function $F$, unless $\W[1] = \FPT$.
\end{abstract}
\newpage

% !TeXroot = main.tex
\section{Introduction}
Imagine inviting a group of friends to test several multi-course menus.
Each menu consists of a fixed sequence of courses, and you have prepared one menu for each friend.
The only decision left is which friend should receive which menu.
There is, however, a deadline: an event that you all want to attend starts soon.
Moreover, eating a course takes a different amount of time for different friends, and you can only begin the next course once everyone has finished the current one.
How should the menus be assigned to the friends so that the entire dinner is completed as early as possible? See~\Cref{tab:menu} for an example.

{
\newcommand{\dotrect}[2]{
\begin{tikzpicture}[baseline={(box.base)}]
	\node[
	draw,                           
	pattern={Dots[distance=1.5mm, radius=.125mm]},   
	minimum width=#1mm,              
	minimum height=.5cm             
	] (box) {#2};
\end{tikzpicture}
}
\newcommand{\hatchrect}[2]{
\begin{tikzpicture}[baseline={(box.base)}]
	\node[
	draw,                           
	pattern={Hatch[distance=2mm, angle=45, line width=0.1pt]}, 
	minimum width=#1mm,              
	minimum height=.5cm              
	] (box) {#2};
\end{tikzpicture}
}
\newcommand{\linerect}[2]{
\begin{tikzpicture}[baseline={(box.base)}]
	\node[
	draw,                           
	pattern={Lines[distance=2mm, angle=45, line width=0.1pt]},  
	minimum width=#1mm,              
	minimum height=.5cm              
	] (box) {#2};
\end{tikzpicture}
}

\begin{table}[h]
	\caption{Top: A possible instance of menus and eating times. Each rectangle represents a dish and the fill patterns of the rectangles represent the courses. Bottom: The optimal solution of cost~6. Each course takes 2 time units.}
	\label{tab:menu}
	\centering
	\begin{tabular}{l|ccc}
		\textbf{Instance} & Menu 1 & Menu 2 & Menu 3\\\hline
		Friend 1\rule[-5pt]{0pt}{17pt} & \dotrect{8}{2}\hatchrect{4}{1}\linerect{12}{3} &\dotrect{8}{2}\hatchrect{8}{2}\linerect{8}{2}&\dotrect{8}{2}\hatchrect{8}{2}\linerect{12}{3}\\%\hline
		Friend 2\rule[-5pt]{0pt}{17pt} &\dotrect{4}{1}\hatchrect{4}{1}\linerect{8}{2} &\dotrect{8}{2}\hatchrect{4}{1}\linerect{12}{3}&\dotrect{8}{2}\hatchrect{4}{1}\linerect{12}{3}\\%\hline
		Friend 3\rule[-5pt]{0pt}{17pt} & \dotrect{4}{1}\hatchrect{12}{3}\linerect{8}{2} &\dotrect{8}{2}\hatchrect{4}{1}\linerect{4}{1}&\dotrect{8}{2}\hatchrect{8}{2}\linerect{8}{2}%\\\hline
	\end{tabular}
	\medskip

	\begin{tabular}{l|lll}
	\textbf{Solution}& Course 1 & Course 2 & Course 3\\\hline
	Friend 1 (Menu 2)\rule[-5pt]{0pt}{17pt} & \dotrect{10}{2}&\hatchrect{10}{2}&\linerect{10}{2}\\
	Friend 2 (Menu 1)\rule[-5pt]{0pt}{17pt} &\dotrect{5}{1}&\hatchrect{5}{1}&\linerect{10}{2} \\
	Friend 3 (Menu 3)\rule[-5pt]{0pt}{17pt} & \dotrect{10}{2}&\hatchrect{10}{2}&\linerect{10}{2}
	\end{tabular}
\end{table}

}

This setting can be identified with a version of the job-shop problem.
In the job-shop scheduling problem (JSSP), we are given $n$ jobs, each consisting of $m$ operations $O_1, \dots, O_m$ with varying processing times, and a station with $m$ machines with varying processing power.
The task is to schedule the jobs on the machines to minimize the makespan, i.e., the time when the last operation finishes.
In our setting, however, we are given $n$ \emph{synchronized} stations, each with a single machine\footnote{We could also assume $m$ machines in a station with operation $i$ running on machine $i$ only.},
on which the $i$-th operation can only start once all machines have finished the $(i-1)$-st operation. The task is to distribute the jobs over all $n$ stations, i.e., we want to find a matching between jobs and stations that minimizes the makespan.
We obtain the following matching decision problem (see~\Cref{fig:matching} for an example).

\begin{problem}[{Sum of Bottleneck Matching Problem}, \MSBM]~%
\begin{description}
	\item[Given:] ~%
	 \begin{itemize}
	 	\item A bipartite graph $G = (S\,\dot\cup\, T, E)$,
	 	\item a vector-valued cost function $c: E\rightarrow \mathbb{N}_0 ^k, e \mapsto (c_1(e),\ldots,c_k(e))$, $k\in \mathbb N$,
	 	\item an integer $Z\in \mathbb N$.
	 \end{itemize}
	\item[Question:] Is there a perfect matching $M \subseteq E$ such that
	$\sum_{i=1}^k \max_{e\in M} c_i(e) \leq Z$?
\end{description}
\end{problem}

\begin{figure}
	\centering
	\includegraphics{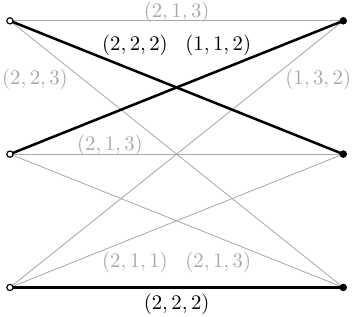}
	\caption{The menu assignment problem from \cref{tab:menu} as an $\MSBM$ instance, with friend nodes on the left and menu nodes on the right. Bold black edges are a solution for $Z=6$.}
	\label{fig:matching}
\end{figure}

This problem is a generalization of the matching problem under categorization constraints studied by Richey and Punnen~\cite{richey1992minimum}.
They provide \NP-hardness results even when restricted to outerplanar graphs; however, no parameterized complexity analysis is given.

This matching problem also arises in the context of reconfiguration problems, where movement is restricted to a single direction per step.
This is the case for tilt problems~\cite{avilajimenez2026generalcomputation,caballero2020hardness,caballero2021fast,caballero2025uniform}, where a set of passive tiles is globally controlled (e.g., by magnetic or gravitational forces being applied to all tiles simultaneously), or for linear-strength robots that are able to push a whole line of robots in a single move~\cite{almethen2020pushing,almethen2022connectedpushing,almethen2023distributed}.
In the latter problems, the authors of these works considered pushing either a single line in each step or multiple different lines with no direction restriction, but an intermediate variant could allow pushing an arbitrary number of lines in a fixed direction per step to avoid complicated collision prevention.
The general task in reconfiguration problems is to move these tiles or robots from a given start configuration to a given target configuration.
A common technique is to compute a bottleneck matching between start and target positions to obtain a lower bound on the number of parallel controls needed;
see, e.g., \cite{fekete2023connected, fekete2024efficiently} for coordinated robot motion planning.
However, no work has considered adapting this approach to models where different agents cannot move in different directions simultaneously.
Therefore, to compute a better lower bound, instead of a single cost value, we could assign to each edge in the start-target graph a tuple of four values for the required up, down, left, and right movement.
We then want to find a perfect matching of minimum cost, resulting in the following natural minimization variant of \MSBM.

\begin{problem}[{Minimum Sum of Bottleneck Matching Problem}, \minMSBM]~%
	\begin{description}
		\item[Given:] ~%
		\begin{itemize}
			\item A bipartite graph $G = (S\,\dot\cup\, T, E)$,
			\item a vector-valued cost function $c: E\rightarrow \mathbb{N}_0 ^k, e \mapsto (c_1(e),\ldots,c_k(e))$, $k\in \mathbb N$.
		\end{itemize}
		\item[Wanted:] A perfect matching $M\subseteq E$ such that
		$\sum_{i=1}^k \max_{e\in M} c_i(e)$ is minimum.
	\end{description}
\end{problem}

In this paper, we consider both the decision and the optimization versions of this matching problem and show the following results.

\begin{description}
	\item[{Fixed-parameter tractability}] \MSBM{} can be decided in time $O((Z+1)^k n^{2.5})$, i.e., it is in $\FPT$ when parameterized by~$k$ and~$Z$ simultaneously (\cref{lem:fpt}).
	\item[{\W[\P]-membership}] \MSBM{} is in $\W[\P]$ when parameterized by either $Z$ or $k$ (\cref{th:wp-membership-k,th:wp-membership-Z}),
                                   i.e., it can be decided in \FPT-time by making at most $O(g(Z) \log n)$ or $O(g(k) \log n)$ non-deterministic guesses for some function $g$.
	\item[{\W[\SAT]-hardness}] \MSBM{} is $\W[\SAT]$-hard when parameterized by either $Z$ or $k$,
                                   even on planar graphs (\cref{th:wsatZ,th:wsatK}), and we thus cannot hope for an \FPT-time algorithm unless $\FPT=\W[\SAT]$.
	\item[{\paraNP-hardness}]  Unless $\P=\NP$, $\MSBM$ parameterized by the maximum degree $\Delta$ of the input graph is not in $\XP$,
                                   i.e., it cannot be decided in time $O(f(\Delta)\cdot n^{g(\Delta)})$ for computable functions~$f$ and~$g$ (\Cref{th:maxDegree}).
                                   Similarly, unless $\P=\NP$, $\MSBM$ parameterized by the treewidth~$\tw$ of the input graph is not in $\XP$ (\cref{cor:tw}).
	\item[Approximability]     For each constant integer $1 \leq d \leq k$, we provide a polynomial-time $\nicefrac k d$\nobreakdash-approximation for \minMSBM{} (\cref{th:k-approx}).
                                   Furthermore, we provide an efficient parameterized approximation scheme (EPAS),
                                   i.e., a $(1+\varepsilon)$-approximation scheme running in \FPT-time with parameter $k + \nicefrac{1}{\varepsilon}$ (\cref{th:approx-scheme}).
	\item[Approximation bounds] Unless $\P = \NP$, \minMSBM{} is hard to approximate within a factor of $k^{\nicefrac{1}{(\log \log k)^{O(1)}}}$, even on planar graphs (\cref{cor:approx-hard}), and, unless $\W[1] = \FPT$, there is no $\FPT$-approximation algorithm for $\minMSBM$ parameterized by $Z$ (\cref{th:noEpasZ}).
\end{description}

\subsection{Related work}
A sum-of-bottlenecks variant of the assignment problem has been studied in~\cite{aggarwal1986bottleneck,punnen2004bottleneck}.
Here, $n$ jobs have to be distributed to $m$ machines, and the objective is to minimize the sum, over all machines, of the longest jobs on each machine.
Heuristic algorithms and \NP-hardness results are provided for this type of problem.
However, the jobs are not synchronized across steps, and no parameterized complexity analysis has been done so far.

The matching variant closest to our problem has been studied in~\cite{richey1992minimum}, where the edges of a graph are partitioned into $k$ categories
and the objective is to find a matching in the graph such that the sum of bottleneck values over all categories is minimized.
This matches our vector-based definition but imposes the additional restriction that for each edge, all entries except one are 0.
Other work considered the same constraints for finding shortest paths~\cite{averbakh1994categorized}, a TSP tour~\cite{punnen1992traveling}, spanning trees~\cite{richey1992minimum}, and minimum cuts~\cite{hassin2007complexity}.
However, only heuristic approaches and \NP-hardness results were given;
we provide a thorough analysis of parameterized complexity and (in-)approximability.

Matchings have been researched for decades in numerous variants.
Thus, we can only give a brief overview.
For a more detailed survey of different algorithms and their comparison, refer to~\cite{2014_Duan_Pettie}.
The classic maximum-cardinality matching problem has been considered for bipartite graphs, e.g., \cite{1973_Hopcroft_Karp,1997_Goldberg_Kennedy,2004_Goldberg_Karzanov}, and for general graphs, e.g., \cite{1965_Edmonds,1980_Micali_Vazirani}.
By adding edge weights, the task often becomes finding a maximum matching of minimum total weight, which can also be computed in polynomial time, see, e.g.,~\cite{1955_Kuhn,1989_Gabow_Tarjan,1997_Goldberg_Kennedy}.
Closer to our scenario is the bottleneck matching problem, where the weight of the heaviest edge in the matching is to be minimized.
Without further constraints, this problem is still solvable in polynomial time~\cite{1988_Gabow_Tarjan}, even in the Euclidean plane~\cite{1992_Chang_Tang_Lee,2000_Efrat_Katz}. However, if one requires non-crossing edges in geometric settings, this problem becomes \NP-hard, but admits a constant-factor approximation~\cite{2014_Katz_et_al}.

Another closely related problem is the Generalized Assignment Problem (GAP), where jobs have to be assigned to capacity-constrained agents while minimizing the total resource consumption~\cite{1975_Ross_Soland}.
The GAP is \NP-hard \cite{1986_Fischer_Jaikumar_Van_Wassenhove}.
However, variations of this problem have been closely studied by many authors, e.g., \cite{1992_Catrysse_Van_Wassenhove}, as the GAP finds use in various scenarios such as vehicle routing \cite{1981_Fischer_Jaikumar} or the facility location problem \cite{1977_Ross_Soland}.

Variants of the GAP include the Task Bottleneck GAP, where the most costly job over all agents is to be minimized, or the Agent Bottleneck GAP, where the maximum resource consumption of the agents is to be minimized.
An algorithm for the first variant was developed by Mazzola and Neebe~\cite{1993_Mazzola_Neebe}.
Relaxations and heuristic approaches have been considered by Martello and Toth~\cite{1995_Martello_Toth} in 1995.
For the latter problem, the objective function is to minimize $\max_{(\text{agent } i)} \sum_{(\text{job } j)} c_{ij}x_{ij}$, whereas the objective in our matching problem is to minimize $\sum_{(\text{agent } i)} \max_{(\text{job } j)} c_{ij}x_{ij}$.
Here, $c_{ij}$ denotes the cost of assigning job~$j$ to agent~$i$, and~$x_{ij}$ denotes the binary choice of whether job~$j$ is assigned to agent~$i$.
A multi-resource variant of the GAP in which each agent has different resources that must not deplete is considered in~\cite{1991_Gavish_Pirkul}.
For this problem, several heuristic approaches exist~\cite{2001_Mazzola_Wilcox,2004_Yagiura_et_al,2009_Mitrovic-Minic_Punnnen}.
A~bottleneck version of the multi-resource GAP has been considered in~\cite{2012_Karsu_Azizoglu,2019_Bektur}.

On the scheduling side, we consider synchronized versions of the job-shop problem.
For a general overview of variants, see, e.g.,~\cite{XIONG2022105731,AGNETIS2025367}.
In the batched job-shop variant, operations of jobs can be organized in batches, and all operations of a batch complete only once its longest operation does~\cite{POTTS2000228,KNOPP201750}.
Makespan minimization on these machines is solvable in polynomial time, but minimizing the sum of completion times is only known to be solvable in time $O(n^{b(b-1)})$ for $n$ jobs and batch size $b$~\cite{brucker1998scheduling}.
This batching process can be seen as a synchronizing variant.
Miller and Davis~\cite{MILLER1989179} consider synchronized finishing times of batches and analyze two scheduling strategies:
(1) the operations do not allow wait times, (2) the operations must start synchronously.
The latter version is closely related to our problem.

% !TeXroot = main.tex
\section{Preliminaries}
In this section we provide basic definitions and complexity classes considered in this paper.
Let $G = (V,E)$ be a graph. A set $M\subseteq E$ is a \emph{matching} if $e\cap f = \emptyset$ for all $e,f\in M$ with $e \neq f$.
A matching $M$ is called \emph{perfect} if every node has an incident edge in $M$.
With the cost function $c: E\rightarrow \mathbb{N}_0 ^k, e \mapsto (c_1(e),\ldots,c_k(e))$, $k\in \mathbb N$, we refer to $c(e)$ as the \emph{cost vector} for an edge $e\in E$.
By \emph{cost entry} we refer to a single value $c_i(e)$ of the cost vector with $i\in [k] = \{1,\dots, k\}$.
For a matching $M$, the vector $\zeta(M) = (\zeta_1,\dots, \zeta_k)$ with
$\zeta_i = \max_{e\in M} c_i(e)$ is called the \emph{bottleneck vector} of $M$.
The \emph{cost $c(M)$} is defined by $c(M)=\sum_{i\in [k]} \zeta_i$.

A parameterized problem is a language $L\subseteq \Sigma^*\times \mathbb N$, where $\Sigma$ is a finite alphabet.
The second component is called the parameter of the problem and generally describes the size of a specific structure.
For our problem, we mainly consider two different parameters: (1)~the length $k$ of the cost vector, and (2)~the maximum cost $Z$ of the matching.

The complexity class $\FPT$ (resp.\ $\XP$) consists of all languages $L$ whose instances $(x,p)$ can be recognized in time $f(p)\cdot |x|^c$, for some computable function $f:\mathbb N\rightarrow \mathbb N$ and a constant $c\in \mathbb N$ (resp.\ in time $f(p)\cdot |x|^{g(p)}$ for computable functions $f$ and $g$).

In this paper, we also consider the complexity classes $\W[\SAT]$ and $\W[\P]$;
for a comprehensive treatment of parameterized complexity, we refer to~\cite{DowneyFellows2013}.
We provide membership and hardness results.
To prove hardness, we use \FPT-reductions.
A parameterized problem $L$ \FPT-reduces to another parameterized problem $L'$ if and only if we can map each instance $(x,p)\in \Sigma^* \times \mathbb N$ to an instance $(x',p')\in \Sigma^* \times \mathbb N$ such that
\begin{itemize}
	\item $(x,p)\in L$ if and only if $(x',p')\in L'$,
	\item $p'\leq g(p)$ for some computable function $g$, and
	\item the reduction is computable in time $f(p) \cdot |x|^c$ for some computable function $f$ and constant $c$.
\end{itemize}
A parameterized problem $L$ is called $\mathcal C$-hard for some parameterized complexity class $\mathcal C$ if every problem in $\mathcal C$ can be $\FPT$-reduced to $L$.
Like polynomial-time reducibility, $\FPT$-reducibility is transitive, implying that
if a parameterized problem $L'$ is $\mathcal C$-hard and can be $\FPT$-reduced to a parameterized problem $L$, then $L$ is also $\mathcal C$-hard.

The complexity class $\W[\SAT]$ contains all parameterized problems that can be $\FPT$-reduced to the weighted Boolean satisfiability problem with parameter $p$.
In this problem, we are given a Boolean formula and ask if we can set $p$ variables to true to satisfy the formula.
If the number of alternations of conjunctions and disjunctions can be bounded by some constant $t\in \mathbb N$, then the class is called $\W[t]$.
The complexity class $\W[\P]$ consists of all parameterized problems that are $\FPT$-reducible to the weighted circuit satisfiability problem with parameter $p$.
Here, a Boolean circuit with $n$ inputs is given, and we ask if we can set $p$ inputs to true so that the circuit outputs true.
Alternatively, $\W[\P]$-problems can be decided by a non-deterministic Turing machine that has at most $O(g(p)\cdot \log n)$ many non-deterministic choices and a total runtime of $O(f(p)\cdot n^c)$, where $f$ and $g$ are computable functions and $c$ is a constant~\cite{FlumGrohe2006}.

\begin{theorem}\label{th:wp-membership-k}
	\MSBM{} is in $\W[\P]$ when parameterized by the length $k$ of the cost vectors.
\end{theorem}
\begin{proof}
	Let $G=(V,E)$ be the input graph.
	For each  $i \in [k]$, we guess one of $O(n^2)$ many edges representing the bottleneck edge for the $i$-th cost entry.
	This requires $O(k \log n)$ bits in total.
	Let $E_B$ be the set of edges guessed and let $\zeta=(\zeta_1,\dots, \zeta_k)$ be the guessed bottleneck vector with $\zeta_i = \max_{e\in E_B}(c_i(e))$.
	If $\sum_{i\in [k]} \zeta_i > Z$, then the guess was incorrect and we reject.
	Otherwise, add to $E_B$ every edge $e\in E$ with $c_i(e) \leq \zeta_i$ for all $i\in[k]$.
	If $G'=(V,E_B)$ has a perfect matching, we accept.
	Otherwise, we reject.

	By construction, we obtain a matching of cost at most $Z$ if the algorithm accepts.
	For the other direction, suppose there is a matching $M$ with cost at most $Z$.
    This cost is determined by a set $E_B$ of at most $k$ edges, and all other edges of $M$ must have a cost vector
    with entries that do not exceed the maximum corresponding entry among edges in $E_B$.
    Thus, all edges of $M$ are contained in $E_B$ and our algorithm accepts.
\end{proof}

\begin{theorem}\label{th:wp-membership-Z}
	\MSBM{} is in $\W[\P]$ when parameterized by the matching cost $Z$.
\end{theorem}
\begin{proof}
	This proof is similar to the proof of \cref{th:wp-membership-k} for parameter $k$.
    If $Z \geq k$, we proceed analogously.
	Otherwise, if $Z < k$, the bottleneck vector of a matching of cost at most~$Z$ has at most~$Z$ non-zero entries.
	Thus, it suffices to guess $Z$ edges.
\end{proof}

% !TeXroot = main.tex
\section{\texorpdfstring{\boldmath $\W[\SAT]$}{W[SAT]}-hardness parameterized by the matching cost \texorpdfstring{\boldmath $Z$}{Z}}
In this section, we prove that $\MSBM$ is $\W[\SAT]$-hard when parameterized by the matching cost $Z$.
The idea is to inductively transform a Boolean formula that uses AND- and OR-operators into a planar graph.
In particular, we are given an instance of \WMFS{} (WMFS), i.e., a Boolean formula~$F$ with $n$ variables containing no negation, and a number $q \in \mathbb{N}$.
	We ask if there is a satisfying assignment with at most $q$ variables set to true.
	Parameterized by~$q$, this problem is $\W[\SAT]$-complete~\cite{AbrahamsonDowneyFellows1995}.
	We may assume that $F$ is a disjunction of conjunctions of disjunctions and so on, alternating.
	Otherwise, we may treat the conjuncts individually and combine the resulting graphs as disjoint components of the constructed~graph.

	We construct a planar bipartite graph~$G_F$ by induction over the structure of a subformula~$\phi$.
	See~\Cref{fig:gadgets} for the recursive structure, and \Cref{fig:reductionZ} for a full example.
	\begin{description}
		\item[Literal] If $\phi$ is a literal $x_i$, we add nodes $s_\phi$ and $t_\phi$, and an edge $\{s_\phi, t_\phi\}$ with cost $e_i$, i.e., the unit vector where the $i$-th entry is $1$ and all other entries are $0$.
		\item[Conjunction] If $\phi$ is a conjunction of formulas $\phi_1 \land \cdots \land \phi_\ell$, we identify $s_\phi = s_{\phi_1}$ and $t_\phi = t_{\phi_\ell}$.
		We add edges $\{t_{\phi_i}, s_{\phi_{i+1}}\}$ for all $1\leq i< \ell$ with cost $0$ in all entries. (See~\Cref{fig:gadgets} left.)
		\item[Disjunction] If $\phi$ is a disjunction of formulas $\phi_1 \lor \cdots \lor \phi_\ell$, we add nodes $s_\phi$, $t_\phi$, as well as $s_\phi^i$ and $t_\phi^i$ for $1\leq i\leq \ell$.
		We add edges $\{s_{\phi}, t_{\phi}^i\}$ and $\{s_{\phi}^i, t_{\phi}\}$, as well as $\{ s_{\phi_{i}}, t_{\phi}^i\}$ and $\{s_{\phi}^i, t_{\phi_{i}}\}$ for all $1\leq i\leq \ell$ with cost $0$ in all entries. (See~\Cref{fig:gadgets} right.)
	\end{description}

	\begin{figure}[p]
		\centering
		\includegraphics[page=2]{./figures/SATtoMBSM-gadgets}\hfil
		\includegraphics[page=1]{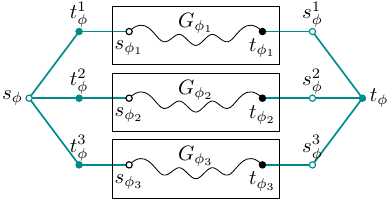}
		\caption{Examples of the graph construction for a given formula $\phi$. Left: $\phi = \phi_1\land \phi_2\land\phi_3$. Right: $\phi = \phi_1\lor \phi_2\lor\phi_3$.}
		\label{fig:gadgets}
	\end{figure}

	\begin{figure}[p]
		\centering
		\includegraphics{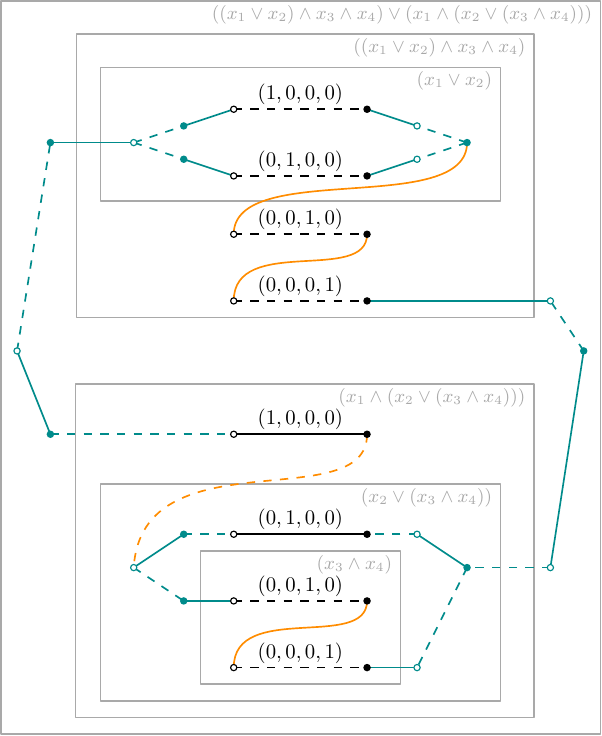}
		\caption{%
			The \MSBM{} instance graph~$G_F$ resulting from the transformation of the WMFS formula $F = {((x_1\lor x_2)\land x_3\land x_4)\lor (x_1\land (x_2\lor (x_3\land x_4)))}$ in the proof of \cref{th:wsatZ}.
			The nodes and edges of literals are black, those of conjunctions are orange, and those of disjunctions are green.
			The non-dashed edges highlight a matching of cost~$2$, which corresponds to a WMFS solution with only~$x_1$ and~$x_2$ set to true.
		}
		\label{fig:reductionZ}
		\end{figure}

	\begin{lemma}\label{lem:reductionruntime}
		Given a \textup{WMFS} instance $F$, $G_F$ can be constructed in polynomial time in $|F|$.
	\end{lemma}

	\begin{proof}
		We have $\mathsf L + \mathsf O\leq |F|$, where $\mathsf L$ ($\mathsf O$, resp.) is the number of literals ($\land$ and $\lor$ operators, resp.) in $F$.
		Each literal introduces two nodes and one edge.
		Given an operator, we can bound the number of nodes and edges added to the graph as follows:
		For each conjunction, we add one edge per operator.
		For each disjunction, we have $\ell$ subformulas with $\ell-1$ operators.
		Therefore, per operator, we add at most six nodes ($s_\phi$, $t_\phi$, $s_\phi^i$, and $t_\phi^i$ with $1 \leq i \leq \ell$) and eight edges ($\{s_{\phi}, t_{\phi}^i\}$, $\{s_{\phi}^i, t_{\phi}\}$, $\{ s_{\phi_{i}}, t_{\phi}^i\}$, and $\{s_{\phi}^i, t_{\phi_{i}}\}$ with $1 \leq i \leq \ell$).
		In total, the constructed graph contains $O(\mathsf{L} + \mathsf{O})\subseteq O(|F|)$ nodes and edges implying polynomial time for the construction.
	\end{proof}

	Before we prove that our construction is the desired reduction to show $\W[\SAT]$-hardness, we list some properties of $G_F$.
	We start with the fact that $G_F$ is a two-terminal series-parallel graph.
	Following the definition of Eppstein~\cite{eppstein1992parallel}, a two-terminal series-parallel graph is constructed inductively:
	\begin{description}
		\item[Base] The graph with two vertices $s$ and $t$ and a single edge $\{s,t\}$ is a two-terminal series-parallel graph. This graph is also called $P_2$.
		\item[Parallel Composition] A pair of two-terminal series-parallel graphs $(X, Y)$ with terminal pairs $(s_X, t_X)$ and $(s_Y, t_Y)$, respectively, forms a new two-terminal series-parallel graph ${G=P(X,Y)}$ through identifying $s=s_X=s_Y$ and $t=t_X=t_Y$.
		\item[Series Composition] A pair of two-terminal series-parallel graphs $(X, Y)$ with terminal pairs $(s_X, t_X)$ and $(s_Y, t_Y)$, respectively, forms a new two-terminal series-parallel graph ${G=S(X,Y)}$ through identifying $s=s_X$, $t_X=s_Y$, and $t=t_Y$.
	\end{description}

	\begin{lemma}\label{lem:spgraph}
		Given a \textup{WMFS} instance $F$, $G_F$ is a two-terminal series-parallel graph.
	\end{lemma}

	\begin{proof}
		To show this, we prove that the construction above can be implemented via a sequence of series and parallel compositions.
		For the literal construction, this is exactly the base case for two-terminal series-parallel graphs.

		Let $\phi$ be a conjunction $\phi_1 \land \cdots \land \phi_\ell$.
		Then we observe that a series composition $S_i = S(G_{\phi_i}, P_2)$ for each $1\leq i < \ell$ followed by a series composition of all resulting graphs and $G_{\phi_\ell}$, which yields $G_\phi$.
		Now, let $\phi$ be a disjunction $\phi_1 \lor \cdots \lor \phi_\ell$.
		We first do a series composition of a three-node path ($P_3$) with each $G_{\phi_i}$ and another $P_3$, i.e., we compute $S_i = S(S(P_3, G_{\phi_i}), P_3)$ for each $1\leq i\leq \ell$.
		Then, we do a parallel composition of all~$S_i$, resulting in~$G_\phi$.
	\end{proof}

	\begin{lemma}\label{lem:planar}
		Given a \textup{WMFS} instance $F$, $G_F$ is bipartite and planar.
	\end{lemma}
	\begin{proof}
		We prove the bipartite property via induction over formula depth, while keeping the invariant $s_\phi\in S$ and $t_\phi\in T$ for all subformulas $\phi$.
		For the base case, the literal construction, we set $s_\phi\in S$ and $t_\phi \in T$, yielding a bipartite partition.
		If $\phi$ is a conjunction, we only add edges between $S$-nodes and $T$-nodes by construction. Thus, the whole construction is still bipartite.
		For the remaining case, consider $\phi$ to be a disjunction.
		By adding $s_\phi$ to $S$, $t_\phi$ to $T$, and also $t_\phi^i$ to $T$ and $s_\phi^i$ to $S$ for $1\leq i\leq \ell$, we obtain a valid bipartition again.

		Two-terminal series-parallel graphs are planar~\cite{duffin1965topology}, yielding planarity via~\cref{lem:spgraph}.
	\end{proof}

\begin{theorem}\label{th:wsatZ}
	\MSBM{} is $\W[\SAT]$-hard when parameterized by the matching cost $Z$, even if the graph is planar.
\end{theorem}

\begin{proof}
	By \cref{lem:reductionruntime,lem:planar}, we can transform a given WMFS formula~$F$ in polynomial time into a bipartite, planar graph~$G_F$.
	It remains to show that~$F$ is satisfiable by setting at most~$q$ variables to true if and only if~$G_F$ has a perfect matching of cost at most~$q$.

	If $F$ is satisfiable with at most $q$ true variables, we recursively construct a perfect matching $M$ in $G_F$ of cost at most $q$ as follows:
	Let $G_\phi$ be the subgraph of $G_F$ that corresponds to the graph obtained by the construction above for formula $\phi$ (initially, $\phi = F$).
	First, assume that $\phi$ is satisfied by the assignment (this is the case for $\phi = F$).
	Consider the cases for $\phi$.
	\begin{description}
		\item[Literal] If $\phi$ is a single literal, then choose the single edge for the perfect matching. Thus, all nodes in $G_\phi$ are matched.
		\item[Conjunction] If $\phi$ is a conjunction of formulas $\phi_1 \land \cdots \land \phi_\ell$, find a perfect matching in each of $G_{\phi_1}, \dots, G_{\phi_\ell}$.
		\item[Disjunction] If $\phi$ is a disjunction of formulas $\phi_1 \lor \cdots \lor \phi_\ell$, choose one $\phi_i$ that is satisfied by the assignment, and treat all other subformulas as unsatisfied.
		Then, (1) find a perfect matching in $G_{\phi_i}$, (2) add edges $\{s_\phi, t_\phi^i\}$ and $\{t_\phi, s_\phi^i\}$ to $M$, and (3) add edges $\{t_{\phi}^j, s_{\phi_{j}}\}$ and $\{s_{\phi}^j, t_{\phi_{j}}\}$ for all $1\leq j\leq \ell, j\neq i$ to $M$.
	\end{description}

	Now, assume that $\phi$ is not satisfied by the assignment or is treated as unsatisfied.
	Note that $\phi \neq F$ for this case.
	By the construction above, we observe that if $\phi$ is treated as unsatisfied, then $s_\phi$ and $t_\phi$ are already matched to nodes outside of $G_\phi$.
	Again, consider the cases for~$\phi$.
	\begin{description}
		\item[Literal] If $\phi$ is a single literal, we have nothing to do. Both nodes are already matched to the outside. Thus, all nodes in $G_\phi$ are matched.
		\item[Conjunction] If $\phi$ is a conjunction of formulas $\phi_1 \land \cdots \land \phi_\ell$, then in graphs $G_{\phi_1}$ and $G_{\phi_\ell}$, there is an odd number of nodes that have to be matched.
		Therefore, we must add edges $\{t_{\phi_i}, s_{\phi_{i+1}}\}$ for all $1\leq i< \ell$ to $M$ to obtain a perfect matching.
		\item[Disjunction] If $\phi$ is a disjunction of formulas $\phi_1 \lor \cdots \lor \phi_\ell$, then nodes $t_{\phi}^i$ can only be matched to nodes $s_{\phi_i}$, and nodes $s_{\phi}^i$ can only be matched to nodes $t_{\phi_i}$.
		Thus, we treat all subformulas as unsatisfied.
	\end{description}

	With this construction, we observe that once a formula $\phi$ is treated as unsatisfied, all edges from $M$ inside $G_\phi$ have cost $0$ and do not contribute to the overall cost.
	If an edge in~$M$ contributes to cost entry $i$, then $x_i$ is set to true.
	If $x_i$ is set to false, the construction above would never treat $\phi = x_i$ as satisfied.
	Thus, the cost of the constructed perfect matching is at most $q$.

	For the other direction, assume that the graph~$G_F$ has a perfect matching of cost at most~$q$.
	As we have only binary cost entries, we have at most $q$ different entries $i_1, \dots, i_q$ contributing to the cost.
	We prove that setting the variables $x_{i_1}, \dots, x_{i_q}$ to true yields an assignment satisfying~$F$.

	Consider a subformula~$\phi$.
	Observe that~$G_\phi$ is balanced, i.e., it has the same number of nodes in~$S$ and~$T$, and only its terminals $s_\phi \in S$ and $t_\phi \in T$ have neighbors outside~$G_\phi$.
	Consequently, in any perfect matching~$M$, exactly one of two states occurs:
	(i)~neither terminal is matched outside~$G_\phi$ ($G_\phi$ is \emph{active}), or (ii)~both terminals are matched outside~$G_\phi$ ($G_\phi$ is \emph{bypassed}).
	We show inductively that if~$G_\phi$ is active, then~$\phi$ is true under the assignment above.
	As before, consider the cases for~$\phi$.
	\begin{description}
		\item[Literal]
		If $\phi = x_i$ is a literal and~$G_\phi$ is active, its only edge $\{s_\phi, t_\phi\}$ must belong to~$M$.
		This edge has cost~$1$ in the $i$-th entry, so $x_i$ is true.
		\item[Conjunction]
		If $\phi$ is a conjunction of formulas $\phi_1 \land \cdots \land \phi_\ell$ and~$G_\phi$ is active, none of the connecting edges $\{t_{\phi_i}, s_{\phi_{i + 1}}\}$ can belong to~$M$.
		Therefore, all~$G_{\phi_i}$ are active and, by induction, all~$\phi_i$ are true.
		\item[Disjunction]
		If $\phi$ is a disjunction of formulas $\phi_1 \lor \cdots \lor \phi_\ell$ and~$G_\phi$ is active, $s_\phi$ is matched to exactly one node~$t_{\phi}^a$ and $t_\phi$ is matched to exactly one node~$s_{\phi}^b$.
		For each~$i$, $s_{\phi_i}$ is matched outside of~$G_{\phi_i}$ if and only if~$t_{\phi_i}$ is.
		Thus, $a = b$ and~$G_{\phi_a}$ is active while all other~$G_{\phi_i}$ are bypassed.
		By induction, $\phi_a$ is true and therefore, $\phi$ is true.
	\end{description}

	Finally, for $\phi = F$, neither~$s_F$ nor~$t_F$ can be matched to nodes outside~$G_F$.
	Thus, $G_F$ is active and by the inductive arguments above, $F$ is satisfied.
	Additionally, since the matching has cost at most~$q$, at most~$q$ variables are set to true.
\end{proof}

The hardness above relies on low cost values with many zero entries.
If all cost vectors consist of positive integers, then $Z \geq k$ and by \cref{lem:fpt} in \cref{sec:approx}, the problem is in~$\FPT$ parameterized by~$Z$ alone.

Because two-terminal series-parallel graphs have treewidth at most 2~\cite{brandstadt1999graph}, we obtain the following corollary by~\cref{lem:spgraph}.
\begin{corollary}\label{cor:tw}
	Unless $\P=\NP$, $\MSBM\notin\XP$ when parameterized by the treewidth $\tw$.
\end{corollary}

This is interesting as many problems become fixed parameter tractable, i.e., are in $\FPT$, when parameterized by treewidth.
Another common parameter for graph-based problems is the maximum degree of a vertex.
We can adapt the construction above in such a way that each node has degree at most three.
To achieve this, we only need to adapt the disjunction gadget by interpreting $\phi=\phi_1\lor\dots\lor\phi_\ell$ recursively as $\phi=\phi'_1\lor\phi'_2$ with $\phi'_1=\phi_1\lor\dots\lor\phi_{\lfloor\ell/2\rfloor}$ and $\phi'_2=\phi_{\lfloor\ell/2\rfloor+1}\lor\dots\lor\phi_{\ell}$.
In other words, we reformulate higher-arity disjunctions using only binary disjunctions; see~\cref{fig:gadgetsBinary} for an example with three subformulas.
It remains to show that with binary operators, the maximum degree is at most three.
We prove this by induction over the depth of the formula and keep two invariants for each subformula $\phi$:
\begin{enumerate}
	\item The maximum degree is at most three in $G_\phi$.
	\item The terminals $s_\phi$ and $t_\phi$ have degree at most two in $G_\phi$.
\end{enumerate}
If $\phi$ is a literal, this is trivially true, because $G_\phi$ consists of a single edge.
If $\phi$ is a conjunction of two formulas $\phi_1$, $\phi_2$, only $t_{\phi_1}$ and $s_{\phi_2}$ increase their degree to at most three,
whereas $s_\phi = s_{\phi_1}$ and $t_\phi = t_{\phi_2}$ keep their degree, which is at most two by induction.
If $\phi$ is a disjunction of formulas $\phi_1$ and $\phi_2$, all new nodes introduced have a degree of two.
Because $s_{\phi_1}$, $s_{\phi_2}$, $t_{\phi_1}$, and $t_{\phi_2}$ get one additional edge, their degree increases to at most three.
Therefore, by induction, every degree is at most three, and we get the following theorem.

\begin{theorem}\label{th:maxDegree}
	Unless $\P=\NP$, $\MSBM\notin \XP$ when parameterized by the maximum degree $\Delta$.
\end{theorem}

\begin{figure}
	\centering
	\includegraphics[page=3]{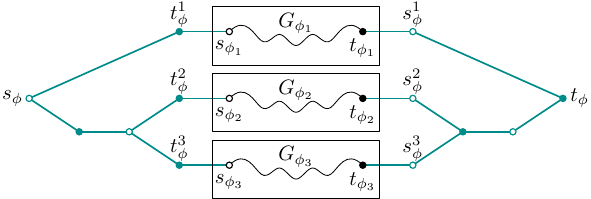}
	\caption{Gadget for a disjunction $\phi=\phi_1\lor\phi_2\lor\phi_3$, expressed using only binary operators.}
	\label{fig:gadgetsBinary}
\end{figure}

% !TeXroot = main.tex
\section{\texorpdfstring{\boldmath $\W[\SAT]$}{W[SAT]}-hardness parameterized by the number of dimensions \texorpdfstring{\boldmath $k$}{k}}
Instead of using $Z$ as a parameter, we now parameterize by the number~$k$ of cost vector entries and prove $\W[\SAT]$-hardness.
The main idea is to encode the choice of picking any variable as the $r$-th true variable into two fixed cost vector entries.

\begin{theorem}\label{th:wsatK}
	\MSBM{} is $\W[\SAT]$-hard when parameterized by the length $k$ of the cost vectors, even if the graph is planar.
\end{theorem}

\begin{proof}
	We use a reduction similar to the one in the proof of \cref{th:wsatZ} with parameter~$Z$.
	Consider a WMFS formula $F$ over $n$ variables with parameter $q$, the number of variables set to true.
	We reduce this to an instance of $\MSBM$ with graph $G_F$, cost function~$c$, maximum cost $Z=qn$, and parameter $k=2q$, the length of the cost vector.
	It suffices to replace the literal gadget.
	We want to add cost $n$ if a variable is the $r$-th true variable.
	If at least two variables $x_i, x_j$ are considered the $r$-th true variable, then the total cost should increase by more than $n$.
	Thus, the matching in the resulting graph has to make sure that no more than $q$ variables are chosen.
	We use the following changes:
	\begin{itemize}
		\item First, replace every literal $x_i$ with the formula $x_{i,1}\lor x_{i,2}\lor \dots \lor x_{i,q}$.
		This clearly does not change the satisfiability of $F$.
		\item When transforming a literal $\phi = x_{i, r}$ into a graph gadget, we construct a path with three edges as follows.
		Add four nodes $s_\phi$ and $t_\phi$, as well as $t'_\phi$ and $s'_\phi$.
		Insert edges $\{t'_\phi, s'_\phi\}$ with cost 0 in all entries, $\{s_\phi, t'_\phi\}$ with cost $c_{2r-1} = i$, and $\{s'_\phi, t_\phi\}$ with cost $c_{2r} = n-i$. All other edge cost entries are 0. The resulting graph $G_F$ can be seen in \Cref{fig:reductionK}.
		\item Add disjoint graphs $G_{F_r}$ for the formula $F_r = (x_{1,r}\lor x_{2,r}\lor \dots\lor x_{n,r})$ for each $r=1,\dots, q$.
		These graphs ensure that up to $q$ variables are chosen (with possible duplicates).
		This corresponds to changing the formula $F$ to $F\land \bigwedge_{r=1}^q F_r$.
		Again, satisfiability of $F$ is not affected by this transformation.
		Due to the graphs $G_{F_r}$, the matching has to choose at least one variable set to true for each $r$, resulting in a cost of at least $qn$.
		See~\Cref{fig:reductionK-chooseGadget} for an example with $n=4$ variables and $q=2$.
	\end{itemize}

	\begin{figure}
		\centering
		\includegraphics[page=2]{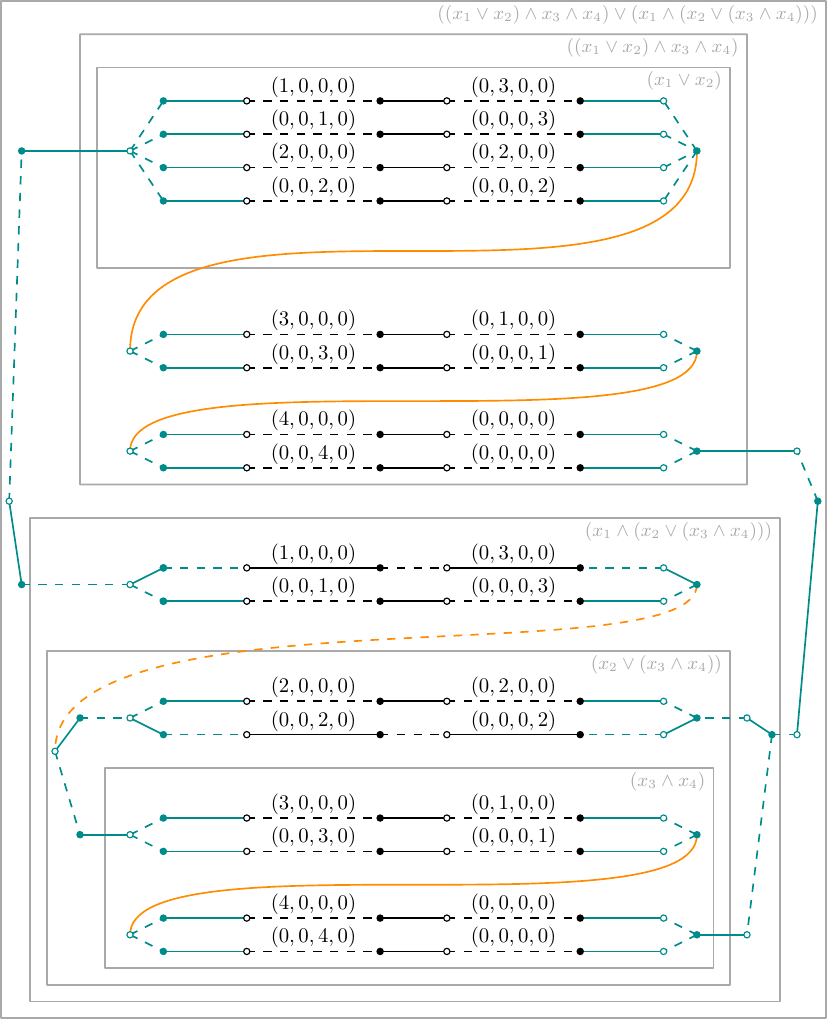}
		\caption{%
			The \MSBM{} instance graph $G_F$ resulting from the transformation of the WMFS formula $F = {((x_1\lor x_2)\land x_3\land x_4)\lor (x_1\land (x_2\lor (x_3\land x_4)))}$ in the proof of \cref{th:wsatK}.
			The nodes and edges of literals are black, those of conjunctions are orange and those of disjunctions are green.
			The non-dashed edges highlight a matching of cost~$8 = (1 + 3) + (2 + 2)$ which corresponds to a WMFS solution with only two true variables ($x_1$ as the first and~$x_2$ as the second true variable).
		}
		\label{fig:reductionK}
	\end{figure}

	\begin{figure}
		\centering
		\includegraphics[page=3]{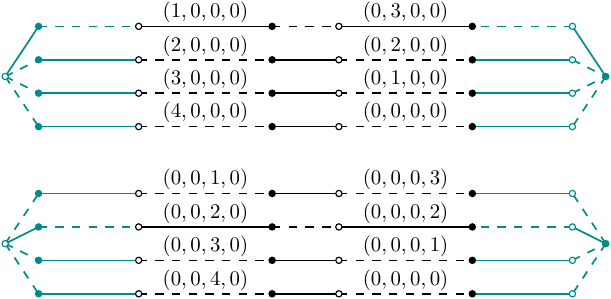}
		\caption{%
			The \MSBM{} instance graphs $G_{F_1}$ and $G_{F_2}$ resulting from the transformation of a WMFS formula with four variables and $q=2$ in the proof of \cref{th:wsatK}.
			The nodes and edges of literals are black, and those of disjunctions are green.
			The non-dashed edges highlight a matching of cost $8 = (1 + 3) + (2 + 2)$ which corresponds to a WMFS assignment with only two true variables ($x_1$ as the first and~$x_2$ as the second true variable).
		}
		\label{fig:reductionK-chooseGadget}
	\end{figure}

	The remaining construction follows the reduction from \Cref{th:wsatZ}.
	Clearly, the resulting graph is again bipartite and planar.
	To prove correctness of the reduction, it is sufficient to consider the literal gadgets.

	In the graphs $G_{F_r}$, we can set only one literal $x_{i, r}$ to true as the $r$-th true variable.
	Now assume that in $G_F$ we choose a different variable, say $x_{j,r}$, to be the $r$-th true variable.
	Then, the bottleneck sum of cost vector entries $2r-1$ and $2r$ will be at least $\max(i,j) + \max({n-i}, {n-j}) = n + |i - j| > n$.
	This results in a total sum of bottlenecks of at least $qn+1$.
	Now assume that we have a matching of cost $qn$.
	For literal gadgets, we have a path of three edges instead of a single edge.
	Consider a formula $\phi$ consisting of a single literal $x_{i,r}$.
	If $G_\phi$ is active, i.e., $s_\phi$ is not matched to a node outside $G_\phi$, we need to exclude the middle edge of the path from the matching.
	Therefore, we add cost $n$ to the matching through cost vector entries $2r-1$ and $2r$.
	With this base case, we can follow the induction step from \Cref{th:wsatZ} to show that if there is a matching of cost $qn$, then the formula is satisfiable.

	Conversely, assume $F$ is satisfiable with $q$ true variables.
	We show that there is a matching of cost exactly $qn$.
	For simplicity, let $x_1,x_2,\dots, x_q$ be the $q$ variables set to true in a satisfying assignment.
	Then, in $G_{F_r}$, we activate the literal gadget for $x_{r,r}$ with cost $n$.
	This gives a matching of cost exactly $qn$.
	For $G_F$, we can follow the proof from \Cref{th:wsatZ}. When we arrive at a literal $x_{i,r}$, we activate the corresponding graph only if $i=r$.
	Because every literal from $F$ was replaced by $x_{i,1}\lor x_{i,2}\lor \dots \lor x_{i,q}$, we can always find one literal gadget we can activate if this formula is satisfied by the assignment.

	It remains to show that the reduction runs in $\FPT$-time.
	Following the proof of~\cref{th:wsatZ}, we only add a constant number of nodes and edges per operator ($\land$ or $\lor$ in $F$).
	For each literal $x_i$ in $F$, however, we introduce $q$ new literals $x_{i,r}$ with $r\in[q]$.
	Therefore, we add $O(q)$ many nodes and edges to the graph for each literal in $F$.
	Additionally, we observe that each graph $G_{F_r}$ contains $6n + 2$ nodes and $7n$ edges.
	Thus, the total size of $G_F$ is in $O(\mathsf O + q\cdot \mathsf L + qn)\subseteq O(q|F| + qn)$, where $\mathsf L$ ($\mathsf O$, resp.) is again the number of literals ($\land$~and~$\lor$ operators, resp.) in $F$.
	This implies $\FPT$-time.
\end{proof}

% !TeXroot = main.tex
\section{(In-)Approximability}
\label{sec:approx}

In the construction for the $\W[\SAT]$-hardness parameterized by~$Z$ (\cref{th:wsatZ}), the number of dimensions of the cost vector for \MSBM{} corresponds to the number of variables of the WMFS instance, i.e., $k = n$.
Additionally, the objective value is preserved exactly in the reduction.
As the optimization variant of WMFS (\textsc{Minimum Monotone Satisfying Assignment}, MMSA) is \NP-hard to approximate within ratio $|F|^{\nicefrac{1}{(\log \log |F|)^{O(1)}}}$~\cite{DinurHarshaKindler2015,DinurSafra2004,Hirahara2022} and as $n \leq |F|$, we get the following approximation hardness for \minMSBM{}.

\begin{corollary}\label{cor:approx-hard}
	Unless $\P = \NP$, \minMSBM{} does not admit a $k^{\nicefrac{1}{(\log \log k)^{O(1)}}}$-approximation algorithm with polynomial runtime, even for planar graphs.
\end{corollary}

This factor is super-logarithmic but sublinear.
Unfortunately, we are unable to show tightness for this result by giving an approximation algorithm achieving this factor.
However, we show that we can achieve an approximation factor in $O(k)$ with an arbitrarily small constant factor in the $O$-notation by searching for perfect matchings in suitably chosen subgraphs of $G$.

\begin{definition}
	Let $M$ be a matching in a graph $G$ with cost function $c$, and consider the bottleneck vector $\zeta(M) = (\zeta_1,\dots, \zeta_k)$ with $\zeta_i = \max_{e\in M} c_i(e)$.
	For an integer $1 \leq d \leq k$ and a matching $M$ (resp.\ a bottleneck vector $\zeta$),
	let $T_d(M)$ (resp.\ $T_d(\zeta)$) be the sum of the $d$ largest values in $\zeta(M)$ (resp.\ $\zeta$).
	By $T_d(G)$ we denote the minimum $T_d(M)$ over all perfect matchings $M$ in $G$;
	$T_d(G)$ is called the \emph{$d$-threshold} of $G$.
\end{definition}

We first observe that the total cost of a matching that realizes the $d$-threshold $T_d$ is at most $\nicefrac{k}{d} \cdot T_d$;
this holds because the mean over all $k$ entries of the bottleneck vector is at most the mean over its $d$ largest entries.
Furthermore, we observe that $T_d(G)$ is a lower bound on the minimum cost \OPT{} of any perfect matching in $G$.
Together, these observations imply that a matching realizing the $d$-threshold provides a $\nicefrac{k}{d}$-approximation for \minMSBM{}.

Next, we prove that we can efficiently decide whether there is a perfect matching $M$ with $T_d(M)$ at most a probe value $\tau_d$, allowing us to use binary search to find $T_d(G)$ and a matching realizing it.
\begin{lemma}\label{lem:approx}
	Given a graph $G$, a cost function $c$ whose entries are at most $N$, an integer $1 \leq d \leq k$, and a target value $\tau_d$,
	we can decide in time $O\left(\binom k d n^{2d + 2.5}\right)$
	whether there exists a perfect matching $M$ in $G$ with ${T_d(M) \leq \tau_d}$.
	Computing $T_d(G)$ and a matching realizing it can be done in
	time $O\left(\binom k d n^{2d + 2.5}\log(dN)\right)$.
\end{lemma}
\begin{proof}
	We search for a matching $M$ with ${T_d(M) \leq \tau_d}$, i.e., an $M$ for which the largest~$d$ entries of $\zeta(M)$ sum to at most $\tau_d$.
	Our algorithm works as follows:
	\begin{enumerate}
		\item Enumerate all choices of entries $i_1, \dots, i_d$ that define $T_d(M)$,
		\item consider all possible combinations for these $d$ entries in $\zeta$ with $\sum_{j\in [d]}\zeta_{i_j} \leq \tau_d$, and
		\item search for a perfect matching in $G$ after removing any edge $e$ that violates the guessed values of $\zeta(M)$,
		      i.e., removing any edge $e$ for which there is an entry $c_i(e) > \zeta_i$ for an $i \in \{i_1, \dots, i_d\}$ or
			  a $c_i(e) > \min \{\zeta_{i_1}, \dots, \zeta_{i_d}\}$ for an $i \notin \{i_1, \dots, i_d\}$;
			  call this graph $G_\zeta$.
	\end{enumerate}

	If the algorithm finds a perfect matching $M$ in $G_\zeta$, then by construction of $G_\zeta$,
	it has ${T_d(M) \leq \tau_d}$ and the algorithm returns~$M$.
	Conversely, if $G$ has a perfect matching $M$ with ${T_d(M) \leq \tau_d}$, let $i_1, \ldots, i_d$ be the indices corresponding to the largest $d$ entries of $\zeta(M)$,
	breaking ties arbitrarily.
	These indices and the corresponding values of $\zeta(M)$ are considered in steps 1 and 2 of the algorithm.
	Thus, the algorithm will consider a graph $G_\zeta$ that contains this perfect matching $M$ and will find it in step 3.

	For step 1, there are at most $\binom{k}{d}$ many different choices.
	For each such choice, there can be at most $(n^2)^d$ many different combinations of values for the $d$ entries.
	Removing costly edges and finding a perfect matching can be done in $O(n^{2.5})$ time~(see~\cite{1973_Hopcroft_Karp} for bipartite graphs and~\cite{1980_Micali_Vazirani} for general graphs).
	In total, the runtime of this algorithm is $O\left(\binom k d n^{2d + 2.5}\right)$.

	Since $T_d(G)$ lies in $[0, dN]$,
	we can apply binary search over $\tau_d$ and repeatedly use this algorithm to find $T_d(G)$ in time $O\left(\binom k d n^{2d + 2.5}\log(dN)\right)$.
\end{proof}

In summary, we obtain the following theorem.
\begin{theorem}\label{th:k-approx}
	For each constant integer $1 \leq d \leq k$ there is a $\nicefrac k d$-approximation for $\minMSBM$.
	The runtime of this algorithm is $O\left(\binom k d n^{2d + 2.5}\log(dN)\right)$, where $N$ is the largest cost vector entry.
\end{theorem}

By slightly changing the three-step algorithm, we obtain an $\FPT$-algorithm for $\MSBM$ when parameterized by $k$ and $Z$ simultaneously.
We set $\tau_d \coloneq Z$ and $d \coloneq k$, and iterate over all $(Z+1)^k$ bottleneck vectors, which include the bottleneck vector of a matching with cost at most $Z$ if and only if such a matching exists.
The runtime of this algorithm is $O((Z+1)^k n^{2.5})$, implying $\FPT$-time.
Therefore, we obtain the following corollary.

\begin{corollary}\label{lem:fpt}
	$\MSBM$ is in $\FPT$ when parameterized by $k$ and $Z$ simultaneously.
\end{corollary}

In contrast to our lower approximation bound, the problem admits an efficient parameterized approximation scheme (EPAS) with parameter $k$.
The fundamental idea of this EPAS is to generate a sufficiently fine grid $\Gamma$ of guesses for each entry of the bottleneck vector
and to search for a perfect matching in the subgraphs obtained by applying the $|\Gamma|^k$ combinations of these guesses as thresholds on the edge costs;
the size of $\Gamma$ depends only on $\varepsilon$ and $k$, not on $n$ or the cost values.

\begin{theorem}\label{th:approx-scheme}
	For any sufficiently small $\varepsilon > 0$, there is a $(1+\varepsilon)$-approximation algorithm for $\minMSBM$ that runs in time $O((C\log(k/\varepsilon)/\varepsilon)^k \cdot kn^{2.5})$ for some constant $C$.
\end{theorem}
\begin{proof}
	Recall that, for a matching $M$, $\zeta(M)$ is its bottleneck vector and $c(M) = \sum_{i \in [k]}\zeta_i(M)$ is its cost.
	For a threshold vector $\theta \in \mathbb{R}^k$, let $G_\theta = (V, E_\theta)$ denote the \emph{threshold graph} with $E_\theta = \{ e \in E \mid \forall i \in [k]\;c_i(e) \leq \theta_i\}$.

	Our algorithm works as follows.
	Given our input graph $G$, we first compute an ordinary minimum bottleneck matching,
        assigning the maximum entry of each vector $\max_{j \in [k]} c_j(e)$ as cost to edge $e$.
	If the graph has no perfect matching, we report that and are done.
	Otherwise, let $B_0$ be the bottleneck value of this ordinary bottleneck matching, i.e.,
	the minimum value for which removing all edges with a higher cost entry still admits a perfect matching.
	If $B_0 = 0$, the bottleneck matching is clearly optimal for $\MSBM$; we return it and are done.

	Otherwise, let $\varepsilon' \coloneq \nicefrac{\varepsilon}{2}$, let $\gamma \coloneq \frac{\varepsilon' B_0}{k}$, and let $J \coloneq \lceil \log_{1+\varepsilon'}(\nicefrac{k^2}{\varepsilon'})\rceil$.
	We define the grid $\Gamma \coloneq \{ \gamma \cdot (1+\varepsilon')^j \mid j \in \mathbb{N}_0, 0 \leq j \leq J\}$.
	For every $\theta \in \Gamma^k$, we build $G_\theta$ and search for a perfect matching $M_\theta$ on it.
	If a perfect matching exists on $G_\theta$, we compute its cost $c(M_\theta)$.
	In the end, we return the cheapest matching $M_\theta$ found among all $\theta \in \Gamma^k$.

	To prove that this algorithm has all claimed properties, we make the following three observations.
	(O1)~Since the optimal solution for $\MSBM$ must contain an edge whose maximum cost entry is at least $B_0$, we have $B_0 \leq \OPT \leq k \cdot B_0$.
	(O2)~From the component-wise inequality $c(e) \leq \theta$ for all $e \in E_\theta$, we have $c(M_\theta) \leq \sum_{i\in [k]} \theta_i$.
	(O3)~If $M$ is a perfect matching and $\theta \geq \zeta(M)$ component-wise, then $G_\theta$ has a perfect matching.

	Observe that, if no perfect matching exists, our algorithm correctly reports this, so we can assume that a perfect matching exists.
	Let $M^*$ be an optimal perfect matching and $\zeta^* \coloneq \zeta(M^*)$ be its bottleneck vector.
	If $M^*$ has cost $\OPT = 0$, by (O1), we also find a matching of cost $0$, so we can now assume $\OPT \geq B_0 > 0$.

	By the choice of $J$, we have $\max \Gamma = \gamma \cdot (1+\varepsilon')^J \geq \gamma \cdot \nicefrac{k^2}{\varepsilon'} = kB_0 \geq \OPT \geq \zeta^*_i$ for all $i$;
	in other words, each entry of $\zeta^*$ is below some grid value from $\Gamma$.
	Together with the existence of $M^*$ and (O3), this guarantees that our algorithm finds and returns a perfect matching.
	Let $\hat{\theta}_i \in \Gamma, \hat{\theta}_i \geq \zeta^*_i$ be the smallest $\Gamma$-value that is at least $\zeta^*_i$, and let $\hat{\theta} \in \Gamma^k$ be the vector of these $k$ values.
	The existence of $M^*$ together with (O3) implies that $G_{\hat{\theta}}$ has a perfect matching $M_{\hat{\theta}}$;
	in the following, we argue that $M_{\hat{\theta}}$ yields the claimed approximation.

	First, we prove that $\hat{\theta}_i \leq (1+\varepsilon')\zeta^*_i + \gamma$ for all $i \in [k]$ by distinguishing two cases.
	In the first case, if $\zeta^*_i \leq \gamma$, then $\hat{\theta}_i = \gamma \in \Gamma$.
	Otherwise, we have $\gamma \cdot (1+\varepsilon')^{j-1} < \zeta^*_i \leq \hat{\theta}_i = \gamma \cdot (1+\varepsilon')^j$ for some $j \geq 1$,
	so we have $\hat{\theta}_i \leq (1+\varepsilon')\zeta^*_i$.
	By (O2), the cost $c(M_{\hat{\theta}})$ is at most the sum of thresholds $\sum_{i \in [k]} \hat{\theta}_i$.
	We thus have
	\[ c(M_{\hat{\theta}}) \leq \sum_{i \in [k]} \hat{\theta}_i \leq (1+\varepsilon')\sum_{i \in [k]} \zeta^*_i + k\gamma = (1+\varepsilon')\OPT + \varepsilon'B_0 \overset{\textup{(O1)}}{{}\leq{}} (1+2\varepsilon')\OPT = (1+\varepsilon)\OPT\text{,}\]
	proving the approximation guarantee.

	Finally, the runtime of our algorithm is dominated by the $|\Gamma|^k$ subgraph constructions and searches for perfect matchings, each of which can be done in $O(kn^{2.5})$ time.
	For $\varepsilon' > 1$, we have $\ln(1+\varepsilon') \in \Omega(1)$; for $\varepsilon' \in (0,1]$, we have $\ln(1+\varepsilon') \geq \varepsilon'/2$.
	We thus have $|\Gamma| = J+1 = \lceil \log_{1+\varepsilon'}\nicefrac{k^2}{\varepsilon'} \rceil + 1 = \lceil \ln(\nicefrac{k^2}{\varepsilon'})/\ln (1+\varepsilon') \rceil + 1 =
        O(\log(k/\varepsilon)/\varepsilon)$.
	Thus, for some sufficiently large absolute constant $C$, our runtime is $O((C\log(k/\varepsilon)/\varepsilon)^k \cdot kn^{2.5})$.
\end{proof}

In contrast to this positive result for parameter $k$, we can rule out any approximation scheme for parameter $Z$ (the target value) via the $\W[\SAT]$-hardness construction of \cref{th:wsatZ}.

\begin{theorem}\label{th:noEpasZ}
	Let $A(Z)$ and $T(Z)$ be any computable functions of $Z$.
	Unless $\W[1] = \FPT$, there is no $A(Z)$-approximation algorithm with runtime $T(Z) \cdot \operatorname{poly}(n,Z)$ for $\minMSBM$.
\end{theorem}
\begin{proof}
	We begin by observing that the reduction from the proof of \cref{th:wsatZ} exactly preserves the optimal value:
	in the forward direction, an assignment with $w$ positive variables is turned into a matching of cost $w$,
	and in the backward direction, a matching of cost $w$ results in an assignment with $w$ positive variables.
	The $q$-\textsc{Dominating Set} problem can be expressed as an instance of $q$-\WMFS{} (as $\bigwedge_{v \in V} \bigvee_{u \in N[v]} x_u$),
	also preserving the value of an optimal solution.
	Thus, via this reduction, an $A(Z)$-approximation algorithm with runtime $T(Z) \cdot p(n,Z)$ for some polynomial $p$
	for $\minMSBM{}$ directly yields an $A(q)$-approximation algorithm with runtime $T(q) \cdot p'(n, q)$ for some other polynomial $p'$ for $q$-\textsc{Dominating Set}.
	However, Karthik et al.~\cite{DBLP:journals/jacm/SLM19} show that such an algorithm can only exist if $\W[1] = \FPT$.
\end{proof}

% !TeXroot = main.tex
\section{Conclusion and future work}
We gave a parameterized complexity analysis for a generalization of a matching problem from~\cite{richey1992minimum}.
Our results include $\W[\P]$-membership and $\W[\SAT]$-hardness when parameterized by either the length of the cost vector or the maximum cost.
On the algorithmic side, we showed that the problem is fixed-parameter tractable when parameterized by both values simultaneously.
Notably, the hardness results also hold for planar graphs with maximum degree three, demonstrating that the hardness does not arise from arbitrary edge crossings or highly connected nodes.

Future work includes investigating the complexity of the problem for non-bipartite graphs in greater detail.
Because we did not use the fact that $G$ is bipartite in~\Cref{th:wp-membership-k,th:wp-membership-Z}, this problem remains in $\W[\P]$ for both parameters.
If the problem is $\W[\P]$-complete, this may be easier to prove for general graphs than for bipartite graphs.

Variants of the objective function or cost function can also be considered:
For reconfiguration problems, the cost values contain geometric information.
Thus, one could require every entry $c_i$ of the cost function to be induced by a metric $d_i$ on the node set, i.e., $c_i(\{u,v\}) = d_i(u,v)$.
As in~\cite{hassin2007complexity}, instead of minimizing the cost of the matching, future work could also investigate maximizing the cost of the matching.

On the approximation side, we showed super-logarithmic approximation hardness and provided a $\nicefrac k d$-approximation, as well as an EPAS with parameter $k$.
Furthermore, we ruled out similar approximation schemes for $Z$.
Future work could address the gap between our approximation algorithm and the lower bound.

One could also consider other graph problems such as minimum spanning trees or minimum $s$-$t$-cuts under the sum-of-bottlenecks objective.
It would be interesting to see whether complexity results carry over to all these variants or if some are harder or easier to solve.
For $s$-$t$-paths, our reduction suggests $\W[\SAT]$-hardness for each of the parameters $k$ and $Z$.
However, using the number of edges in a shortest path as a parameter may also be worth investigating.

\section*{AI disclosure}
We used LLMs such as OpenAI GPT-5.6 Sol in the preparation of this work for idea discussions, language editing, formatting and polishing proofs.
We did not copy any autonomously generated text into our paper, and the paper was entirely written by us.

\bibliographystyle{plainurl}
\bibliography{references}

\end{document}